%% file: lowdegcsp.tex
\documentclass[11pt,letterpaper]{article}

\def\showauthornotes{0}
\def\showtableofcontents{0}
\def\showkeys{0}
\def\showdraftbox{1}
\def\showcolorlinks{1}
\def\usemicrotype{0}
\def\showfixme{1}

\input{preamble/macros}
\input{preamble/macros-lowdegcsp}

\title{Beating the random assignment on \\ constraint satisfaction problems of bounded degree}

\author{%
Boaz Barak\thanks{Microsoft Research New England.}
\and
Ankur Moitra\thanks{MIT Mathematics Department.}
\and
Ryan O'Donnell\thanks{Department of Computer Science, Carnegie Mellon. }
\and
Prasad Raghavendra\thanks{U.C.Berkeley, Department of Electrical Engineering \& Computer Sciences.}
\and
Oded Regev\thanks{Courant Institute of Mathematical Sciences, New York University.}
\and
David Steurer\thanks{Cornell University.}
\and
Luca Trevisan\footnotemark[4]
\and
Aravindan Vijayaraghavan\thanks{Courant Institute of Mathematical Sciences, New York University.}
\and
David Witmer\footnotemark[3]
\and
John Wright\footnotemark[3]
}

\begin{document}

\maketitle
\thispagestyle{empty}

\input{content/abstract}
\clearpage

\ifnum\showtableofcontents=1
{
\tableofcontents
\thispagestyle{empty}
 }
\fi

\clearpage

\setcounter{page}{1}

\input{content/introduction}
\input{content/preliminaries}

\input{content/3xor}

\input{content/dfko}
\input{content/trianglefree}

\section*{Acknowledgments}
We thank Scott Aaronson for bringing the paper of Farhi et al.~\cite{FGG14} to (some of) the authors' attention. RO, DW, and JW were supported by NSF grants CCF-0747250 and CCF-1116594.  DW was also supported by the NSF Graduate Research Fellowship Program under grant DGE-1252522; JW was also supported by a Simons Graduate Fellowship. OR, DS, and AV acknowledge the support of the Simons Collaboration on Algorithms and Geometry. OR was also supported by NSF grant CCF-1320188. DS was also supported by a Sloan fellowship, a Microsoft Research Faculty Fellowship, and by the NSF. Any opinions, findings, and conclusions or recommendations expressed in this material are those of the authors and do not necessarily reflect the views of the NSF.

\addreferencesection
\bibliographystyle{amsalpha}
\bibliography{bib/mr,bib/dblp,bib/scholar,bib/lowdegcsp}



\end{document}

%% file: preamble/macros.tex
\usepackage{etex}

\usepackage[l2tabu, orthodox]{nag}

\usepackage{xspace,enumerate}

\usepackage[dvipsnames]{xcolor}

\usepackage[T1]{fontenc}
\usepackage[full]{textcomp}

\usepackage[american]{babel}

\usepackage{mathtools}

\usepackage{amsthm}

\newtheorem{theorem}{Theorem}[section]
\newtheorem*{theorem*}{Theorem}

\newtheorem{subclaim}{Claim}[theorem]
\newtheorem{proposition}[theorem]{Proposition}
\newtheorem*{proposition*}{Proposition}
\newtheorem{lemma}[theorem]{Lemma}
\newtheorem*{lemma*}{Lemma}

\newtheorem*{conjecture*}{Conjecture}
\newtheorem{fact}[theorem]{Fact}
\newtheorem*{fact*}{Fact}

\newtheorem*{hypothesis*}{Hypothesis}

\theoremstyle{definition}
\newtheorem{definition}[theorem]{Definition}

\theoremstyle{remark}

\newtheorem*{claim*}{Claim}
\newtheorem{remark}[theorem]{Remark}
\newtheorem*{remark*}{Remark}

\newtheorem*{observation*}{Observation}

\usepackage[
letterpaper,
top=0.7in,
bottom=0.9in,
left=1in,
right=1in]{geometry}

\usepackage{mathpazo}
\usepackage{textcomp} %
\usepackage{bm} %
\ifnum\showkeys=1
\usepackage[color]{showkeys}
\fi

\ifnum\showcolorlinks=1
\usepackage[
pagebackref,
colorlinks=true,
urlcolor=blue,
linkcolor=blue,
citecolor=OliveGreen,
]{hyperref}
\fi

\ifnum\showcolorlinks=0
\usepackage[
pagebackref,
colorlinks=false,
pdfborder={0 0 0}
]{hyperref}
\fi

\usepackage{prettyref}

\newcommand{\savehyperref}[2]{\texorpdfstring{\hyperref[#1]{#2}}{#2}}

\newrefformat{eq}{\savehyperref{#1}{\textup{(\ref*{#1})}}}
\newrefformat{lem}{\savehyperref{#1}{Lemma~\ref*{#1}}}
\newrefformat{def}{\savehyperref{#1}{Definition~\ref*{#1}}}
\newrefformat{thm}{\savehyperref{#1}{Theorem~\ref*{#1}}}
\newrefformat{cor}{\savehyperref{#1}{Corollary~\ref*{#1}}}
\newrefformat{cha}{\savehyperref{#1}{Chapter~\ref*{#1}}}
\newrefformat{sec}{\savehyperref{#1}{Section~\ref*{#1}}}
\newrefformat{app}{\savehyperref{#1}{Appendix~\ref*{#1}}}
\newrefformat{tab}{\savehyperref{#1}{Table~\ref*{#1}}}
\newrefformat{fig}{\savehyperref{#1}{Figure~\ref*{#1}}}
\newrefformat{hyp}{\savehyperref{#1}{Hypothesis~\ref*{#1}}}
\newrefformat{alg}{\savehyperref{#1}{Algorithm~\ref*{#1}}}
\newrefformat{rem}{\savehyperref{#1}{Remark~\ref*{#1}}}
\newrefformat{item}{\savehyperref{#1}{Item~\ref*{#1}}}
\newrefformat{step}{\savehyperref{#1}{step~\ref*{#1}}}
\newrefformat{conj}{\savehyperref{#1}{Conjecture~\ref*{#1}}}
\newrefformat{fact}{\savehyperref{#1}{Fact~\ref*{#1}}}
\newrefformat{prop}{\savehyperref{#1}{Proposition~\ref*{#1}}}
\newrefformat{prob}{\savehyperref{#1}{Problem~\ref*{#1}}}
\newrefformat{claim}{\savehyperref{#1}{Claim~\ref*{#1}}}
\newrefformat{relax}{\savehyperref{#1}{Relaxation~\ref*{#1}}}
\newrefformat{red}{\savehyperref{#1}{Reduction~\ref*{#1}}}
\newrefformat{part}{\savehyperref{#1}{Part~\ref*{#1}}}

\newcommand{\Sref}[1]{\hyperref[#1]{\S\ref*{#1}}}

\usepackage{nicefrac}

\ifnum\usemicrotype=1
\usepackage{microtype}
\fi

\ifnum\showauthornotes=1
\newcommand{\Authornote}[2]{{\sffamily\small\color{red}{[#1: #2]}}}
\newcommand{\Authornotecolored}[3]{{\sffamily\small\color{#1}{[#2: #3]}}}
\newcommand{\Authorcomment}[2]{{\sffamily\small\color{gray}{[#1: #2]}}}
\newcommand{\Authorstartcomment}[1]{\sffamily\small\color{gray}[#1: }

\newcommand{\Authorfnote}[2]{\footnote{\color{red}{#1: #2}}}
\newcommand{\Authorfixme}[1]{\Authornote{#1}{\textbf{??}}}
\newcommand{\Authormarginmark}[1]{\marginpar{\textcolor{red}{\fbox{\Large #1:!}}}}
\else
\newcommand{\Authornote}[2]{}
\newcommand{\Authornotecolored}[3]{}
\newcommand{\Authorcomment}[2]{}
\newcommand{\Authorstartcomment}[1]{}

\newcommand{\Authorfnote}[2]{}
\newcommand{\Authorfixme}[1]{}
\newcommand{\Authormarginmark}[1]{}
\fi

\newcommand{\Bnote}{\Authornote{B}}

\ifnum\showfixme=0

\fi

\usepackage{boxedminipage}

\newenvironment{mybox}
{\center \noindent\begin{boxedminipage}{1.0\linewidth}}
{\end{boxedminipage}
\noindent
}

\newcommand{\abs}[1]{\lvert#1\rvert}
\newcommand{\Abs}[1]{\left\lvert#1\right\rvert}

\newcommand{\card}[1]{\lvert#1\rvert}

\newcommand{\set}[1]{\{#1\}}

\newcommand{\norm}[1]{\lVert#1\rVert}

\newcommand{\Esymb}{\mathbb{E}}
\newcommand{\Psymb}{\mathbb{P}}
\newcommand{\Vsymb}{\mathbb{V}}

\DeclareMathOperator*{\E}{\Esymb}
\DeclareMathOperator*{\Var}{\Vsymb}
\DeclareMathOperator*{\ProbOp}{\Psymb}

\renewcommand{\Pr}{\ProbOp}

\newcommand{\suchthat}{\;\middle\vert\;}

\newcommand{\textparen}[1]{\text{(#1)}}

\ifx\because\undefined
\newcommand{\because}[1]{\textparen{because #1}}
\else
\renewcommand{\because}[1]{\textparen{because #1}}
\fi
\newcommand{\by}[1]{\textparen{by #1}}

\newcommand{\sbits}{\{\pm1\}}

\newcommand{\defeq}{\stackrel{\mathrm{def}}=}

\newcommand{\mper}{\,.}
\newcommand{\mcom}{\,,}

\newcommand\bdot\bullet

\ifx\mathds\undefined %
\newcommand{\Ind}{\mathbb I}
\else
\newcommand{\Ind}{\mathds 1}
\fi

\DeclareMathOperator{\Inf}{Inf}

\DeclareMathOperator{\poly}{poly}

\DeclareMathOperator{\sign}{sign}

\newcommand{\R}{\mathbb R}

\newcommand{\cP}{\mathcal P}

\newcommand{\cU}{\mathcal U}

\ifnum\showdraftbox=1

\else

\fi

\let\epsilon=\varepsilon

\numberwithin{equation}{section}

\newcommand\MYcurrentlabel{xxx}

\newcommand{\MYstore}[2]{%
  \global\expandafter \def \csname MYMEMORY #1 \endcsname{#2}%
}

\newcommand{\MYload}[1]{%
  \csname MYMEMORY #1 \endcsname%
}

\newcommand{\MYnewlabel}[1]{%
  \renewcommand\MYcurrentlabel{#1}%
  \MYoldlabel{#1}%
}

\newcommand{\MYdummylabel}[1]{}

\newcommand{\torestate}[1]{%
  \let\MYoldlabel\label%
  \let\label\MYnewlabel%
  #1%
  \MYstore{\MYcurrentlabel}{#1}%
  \let\label\MYoldlabel%
}

\newcommand{\restatetheorem}[1]{%
  \let\MYoldlabel\label
  \let\label\MYdummylabel
  \begin{theorem*}[Restatement of \prettyref{#1}]
    \MYload{#1}
  \end{theorem*}
  \let\label\MYoldlabel
}

\newcommand{\restatelemma}[1]{%
  \let\MYoldlabel\label
  \let\label\MYdummylabel
  \begin{lemma*}[Restatement of \prettyref{#1}]
    \MYload{#1}
  \end{lemma*}
  \let\label\MYoldlabel
}

\newcommand{\restateprop}[1]{%
  \let\MYoldlabel\label
  \let\label\MYdummylabel
  \begin{proposition*}[Restatement of \prettyref{#1}]
    \MYload{#1}
  \end{proposition*}
  \let\label\MYoldlabel
}

\newcommand{\restatefact}[1]{%
  \let\MYoldlabel\label
  \let\label\MYdummylabel
  \begin{fact*}[Restatement of \prettyref{#1}]
    \MYload{#1}
  \end{fact*}
  \let\label\MYoldlabel
}

\newcommand{\restate}[1]{%
  \let\MYoldlabel\label
  \let\label\MYdummylabel
  \MYload{#1}
  \let\label\MYoldlabel
}

\newcommand{\addreferencesection}{
  \phantomsection
  \addcontentsline{toc}{section}{References}
}

\newcommand{\sse}{\subseteq}

\newcommand{\eps}{\epsilon}

\let\origparagraph\paragraph
\renewcommand{\paragraph}[1]{\origparagraph{#1.}}

\allowdisplaybreaks

\usepackage{paralist}

\usepackage{comment}

\usepackage{eufrak}

%% file: preamble/macros-lowdegcsp.tex
\newcommand{\AVnote}{\Authornote{AV}} %

\newcommand{\Onote}{\Authornote{O}} %

\newcommand{\ROnote}{\Authornote{R}} %

\let\pref=\prettyref

\renewcommand{\Vsymb}{\mathrm{Var}}
\newcommand{\stddev}{\mathrm{stddev}}

\renewcommand*{\Ind}{\bm 1}

\newcommand*{\pmo}{\{\pm 1\}}

\let\pref=\prettyref

\renewcommand{\Vsymb}{\mathrm{Var}}

\renewcommand{\Ind}{\bm 1}
\let\1\Ind

\newcommand*{\inst}{\Im}

\newcommand*{\maj}{\mathsf{bmaj}}

\newcommand{\clf}{\mathfrak{P}}
\newcommand{\Infl}{\text{Inf}}

\newcommand{\klin}[1]{Max-$#1$XOR\xspace}

\newcommand{\degree}{D}

\newcommand{\grest}{g_{y}}
\newcommand{\gresth}{\widehat{g}_{y}}

\let\originalleft\left
\let\originalright\right
\renewcommand{\left}{\mathopen{}\mathclose\bgroup\originalleft}
\renewcommand{\right}{\aftergroup\egroup\originalright}

\newcommand{\wh}[1]{\widehat{#1}}
\newcommand{\wt}[1]{\widetilde{#1}}
\newcommand{\bx}{\boldsymbol{x}}

\renewcommand{\by}{\boldsymbol{y}}
\newcommand{\bz}{\boldsymbol{z}}
\newcommand{\bA}{\boldsymbol{A}}
\newcommand{\bF}{\boldsymbol{F}}
\newcommand{\bG}{\boldsymbol{G}}
\newcommand{\bN}{\boldsymbol{N}}
\newcommand{\bQ}{\boldsymbol{Q}}

\newcommand{\littlesum}{\mathop{{\textstyle \sum}}}
\newcommand{\sgn}{\mathrm{sgn}}
\newcommand{\ol}[1]{\overline{#1}} 

%% file: content/abstract.tex
\begin{abstract}
We show that for any odd $k$ and any instance $\inst$ of the \klin{k} constraint satisfaction problem, 
there is an efficient algorithm that finds an assignment satisfying at least a $\tfrac{1}{2} + \Omega(1/\sqrt{\degree})$ fraction of $\inst$'s constraints, where $\degree$ is a bound on the number of constraints that each variable occurs in.
This improves both qualitatively and quantitatively on the recent work of Farhi, Goldstone, and Gutmann (2014), which gave a \emph{quantum} algorithm to find an assignment satisfying a $\tfrac{1}{2} + \Omega(\degree^{-3/4})$ fraction of the equations.

For arbitrary constraint satisfaction problems, we give a similar result for ``triangle-free'' instances; i.e., an efficient algorithm that finds an assignment satisfying at least a $\mu + \Omega(1/\sqrt{\degree})$ fraction of constraints, where $\mu$ is the fraction that would be satisfied by a uniformly random assignment.
 \end{abstract}


%% file: content/introduction.tex

\section{Introduction}

An instance of a Boolean constraint satisfaction problem (CSP) over $n$ variables $x_1, \dots, x_n$ is a collection of \emph{constraints}, each of which is some predicate~$P$ applied to a constant number of the variables. The computational task is to find an assignment to the variables that maximizes the number of satisfied predicates.  In general the constraint predicates do not need to be of the same ``form''; however, it is common to study CSPs where this is the case.  Typical examples include: Max-$k$SAT, where each predicate is the OR of $k$ variables or their negations; \klin{k}, where each predicate is the XOR of exactly $k$ variables or their negations; and Max-Cut, the special case of \klin{2} in which each constraint is of the form $x_i \neq x_j$. The case of \klin{k} is particularly mathematically natural, as it is equivalent to maximizing a  homogenous degree-$k$ multilinear polynomial over~$\sbits^n$.

Given a CSP instance, it is easy to compute the expected fraction $\mu$ of constraints satisfied by a uniformly random assignment; e.g., in the case of \klin{k} we always have $\mu = \frac12$.  Thus the question of algorithmic interest is to find an assignment that satisfies noticeably more than a~$\mu$ fraction of constraints.  Of course, sometimes this is simply not possible; e.g., for Max-Cut on the complete $n$-variable graph, at most a $\frac12 + O(1/n)$ fraction of constraints can be satisfied.%
\footnote{Another trivial example is the \klin{2} instance with the two constraints $x=y$ and $x\neq y$. For this reason we always assume that our \klin{k} instances
do not contain a constraint and its negation.}
However, even when all or almost all constraints can be satisfied, it may still be algorithmically difficult to beat~$\mu$.  For example, H{\aa}stad~\cite{MR2144931} famously proved that for every $\eps > 0$, given a \klin{3} instance in which a $1-\eps$ fraction of constraints can be satisfied,
it is NP-hard to find an assignment satisfying a $\frac12 + \eps$ fraction of the constraints.  H{\aa}stad showed similar ``approximation resistance'' results for Max-$3$Sat and several other kinds of CSPs.

One possible reaction to these results is to consider subconstant~$\eps$.  For example, H{\aa}stad and Venkatesh~\cite{HV04} showed that for every \klin{k} instance with $m$~constraints,  one can efficiently find an assignment satisfying at least a $\frac12 + \Omega(1/\sqrt{m})$ fraction of them.\footnote{In~\cite{HV04} this is stated as an approximation-ratio guarantee: if the optimum fraction is $\frac12 + \eps$ then $\frac12 + \Omega(\eps/\sqrt{m})$ is guaranteed. However inspecting their proof yields the absolute statement we have made.} (Here, and elsewhere in this introduction, the $\Omega(\cdot)$ hides a dependence on~$k$, typically exponential.)
Relatedly, Khot and Naor~\cite{KN08} give an efficient algorithm for \klin{3}
that satisfies a $\frac12 + \Omega(\eps \sqrt{(\log n)/n})$ fraction of constraints whenever the optimum fraction is $\frac12 + \eps$.

Another reaction to approximation resistance is to consider restricted instances. One commonly studied restriction is to assume that each variable's ``degree'' --- i.e., the number of constraints in which it occurs --- is bounded by some~$\degree$.  H{\aa}stad~\cite{MR1761191} showed that such instances are never approximation resistant.
More precisely, he showed that
for, say, \klin{k},
one can always efficiently find an assignment satisfying at least a $\mu + \Omega(1/\degree)$ fraction of constraints.\footnote{The previous footnote applies also to this result.}
Note that this advantage of $\Omega(1/\degree)$ cannot in general be improved, as the case of Max-Cut on the complete graph shows.

One may also consider further structural restrictions on instances.  One such restriction is that the underlying constraint hypergraph be \emph{triangle-free} (see \pref{sec:preliminaries} for a precise definition).  For example, Shearer~\cite{shearer} showed that for triangle-free graphs there is an efficient algorithm for finding a cut of size at least
$
    \frac{m}{2} + \Omega(1) \cdot \sum_i \sqrt{\mathrm{deg}(i)},
$
where $\deg(i)$ is the degree of the $i$th vertex.  As $\sum_i \sqrt{\deg(i)} \geq \sum_i \frac{\deg(i)}{\sqrt{\degree}} = \frac{2m}{\sqrt{\degree}}$ in $m$-edge degree-$\degree$ bounded graphs, this shows that for \emph{triangle-free} Max-Cut one can efficiently satisfy at least a $\frac12 + \Omega(1/\sqrt{\degree})$ fraction of constraints.  Related results have also been shown for degree-bounded instances of Maximum Acyclic Subgraph~\cite{BergerShor}, Min-Bisection ~\cite{alon} and Ordering $k$-CSPs~\cite{GuruswamiZhou, Makarychev13}. 

\subsection{Recent developments and our work}
In a recent surprising development, Farhi, Goldstone, and Gutmann~\cite{FGG14} gave an efficient \emph{quantum} algorithm that, for \klin{3} instances with degree bound~$\degree$, finds an assignment satisfying a $\frac12 + \Omega(\degree^{-3/4})$ fraction of the constraints.
In addition, Farhi et al.\ show that if the \klin{3} instance is ``triangle-free'' then an efficient quantum algorithm can satisfy a  $\frac 12 + \Omega(1/\sqrt{\degree})$ fraction of the constraints.

Farhi et al.'s result was perhaps the first example of a quantum algorithm providing a better CSP approximation guarantee than that of the best known classical algorithm (namely H{\aa}stad's~\cite{MR1761191}, for \klin{3}).  As such it attracted quite some attention.\footnote{As evidenced by the long list of authors on this paper; see also \url{http://www.scottaaronson.com/blog/?p=2155}.} In this paper we show that classical algorithms can match, and in fact outperform, Farhi et al.'s quantum algorithm.

\paragraph{First result: \klin{k}} We will present two results.  The first result is about instances of \klin{k}.

\begin{theorem} \label{thm:kxor}
There is a constant $c= \exp(-O(k))$ and a randomized algorithm running in time $\poly(m,n, \exp(k))$ that, given an instance $\inst$ of \klin{k} with $m$ constraints and degree at most $\degree$, finds with high probability an assignment $x\in \sbits^n$ such that
\begin{equation}\label{eq:firstpartmainthm}
\Abs{ \text{val}_{\inst}(x) - \frac{1}{2} } \geq \frac{c}{\sqrt{\degree}} \; .
\end{equation}
Here $\text{val}_{\inst}(x)$ denotes the fraction of constraints satisfied by $x$.
In particular, for odd $k$, by trying the assignment and its negation,
the algorithm can output an $x$ satisfying
\begin{equation}\label{eq:secondpartmainthm}
\text{val}_{\inst}(x) \geq \frac{1}{2} + \frac{c}{\sqrt{\degree}} \; .
\end{equation}
\end{theorem}

In \pref{sec:3xor} we give a simple, self-contained proof of \pref{thm:kxor} in the special case of \klin{3}.  For higher~$k$ we obtain it from a more general result (\pref{thm:constructive}) that gives a constructive version of a theorem of Dinur, Friedgut, Kindler and O'Donnell~\cite{Dinuretal}.  This result shows how to attain a significant deviation from the random assignment value for multivariate low-degree polynomials with low influences. See \pref{sec:proof}.

We note that the deviation $\Omega(1/\sqrt{\degree})$ in~\eqref{eq:firstpartmainthm} is optimal.
To see why, consider any $\degree$-regular graph on $n$ vertices, and construct a \klin{2} instance~$\inst$ as follows. For every edge $(i,j)$ in the graph we randomly and independently include either the constraint $x_i = x_j$ or $x_i \neq x_j$.
For every fixed $x$, the quantity $\text{val}_\inst(x)$ has distribution $\frac1m \text{Binomial}(m, \frac12)$, where $m = \frac{n\degree}{2}$.  Hence a Chernoff-and-union-bound argument shows that with high probability all $2^n$ assignments will have $|\text{val}_\inst(x) - \frac12| \leq O(\sqrt{n/m}) = O(1/\sqrt{\degree})$.
This can easily be extended to \klin{k} for~$k > 2$.

\paragraph{General CSPs}
As noted earlier, the case of Max-Cut on the complete graph shows that for general CSPs,
and in particular for \klin{2},
we cannot guarantee a positive advantage of $\Omega(1/\sqrt{\degree})$
as in~\eqref{eq:secondpartmainthm}.
In fact, a positive advantage of $\Omega(1/\degree)$ is the best possible, showing that the guarantee of H{\aa}stad~\cite{MR1761191} is tight in general.

A similar example can be shown for Max-$2$SAT: consider an instance with $D^2$ variables and imagine them placed on a $D\times D$ grid. For any two variables in the same row add the constraint $x \vee y$ and for any two variables in the same column add the constraint $\bar x \vee \bar y$. Then each variable participates in $O(D)$ clauses, and it can be verified that the best assignment satisfies $3/4+O(1/D)$ fraction of the clauses.
We do not know if the same holds for Max-$3$SAT and we leave that as an open question.

\newcommand{\NAE}{\textsc{NAE}}
\newcommand{\AAE}{\textsc{AE}}
Sometimes no advantage over random is possible. For instance, consider the following instance with 8 clauses on 6 variables, in which
\emph{any} assignment satisfies exactly $1/2$ of the clauses:
\begin{align*}
   \{ &\NAE(x_1,x_2,x_3) , \\
	   &\AAE(y_1,x_2,x_3) , \AAE(x_1,y_2,x_3) , \AAE(x_1,x_2,y_3) , \\
	&\NAE(x_1,y_2,y_3) , \NAE(y_1,x_2,y_3) , \NAE(y_1,y_2,x_3) ,  \\
	&\AAE(y_1,y_2,y_3)
	\}\; ,
\end{align*}
where $\NAE$ denotes the ``not all equal'' constraint, and $\AAE$ is the ``all equal'' constraint.

\paragraph{Second result: triangle-free instances of general CSPs}
Despite the above examples, our second result shows that it is possible to recover the optimal advantage of $1/\sqrt{\degree}$ for
triangle-free instances of \emph{any CSP}:

\begin{theorem} \label{thm:maingencsp}
  There is a constant $c= \exp(-O(k))$ and a randomized algorithm running in time $\poly(m,n,\exp(k))$ time that, given a triangle-free, degree-$\degree$ CSP instance~$\inst$  with $m$ arbitrary constraints, each of arity between~$2$ and~$k$, finds with high probability an assignment $x \in \sbits^n$ such that
  \[
    \text{val}_\inst(x) \geq  \mu + \frac{c}{\sqrt{\degree}}.
  \]
  Here $\mu$ is the fraction of constraints in~$\inst$ that would be satisfied in expectation by a random assignment.
\end{theorem}
\noindent  This theorem is proved in \pref{sec:trianglefree}. 
For simplicity, we state our results as achieving \emph{randomized} algorithms and leave the question of derandomizing them (e.g., by replacing true random bits with $O(k)$-wise independence or some other such distribution) to future work.

\subsection{Overview of our techniques}

All three algorithms that we present in this work follow the same
broad outline, while the details are different in each case.
To produce an assignment that beats a random
assignment, the idea is to partition the variables in to two sets
$(F,G)$ with $F$ standing for `Fixed' and $G$ standing for `Greedy' (in \pref{sec:proof}, these correspond to $[n] \setminus U$ and $U$ respectively).
The variables in $F$ are assigned independent and uniform random bits
and the variables in $G$ are assigned values \emph{greedily} based on
the values already assigned to $F$.
We will refer to constraints with exactly one variable from $G$ as \emph{active}
constraints.  The design of the \emph{greedy}
assignments and their analysis is driven by two key objectives.
\begin{enumerate}
  \item \label{item:obj1} Obtain a significant advantage over the random assignment on
      \emph{active} constraints.
\item \label{item:obj2} Achieve a value that is at least as good as
  the random assignment on inactive constraints.
\end{enumerate}

The simplest example is the algorithm for \klin{3} that we present
in \pref{sec:3xor}.  First, we appeal to a \emph{decoupling} trick due
to Khot-Naor \cite{KN08} to give an efficient approximation-preserving
reduction from an arbitrary instance $\inst$ of \klin{3} to a
bipartite instance $\tilde{\inst}$.  Specifically, the instance $\tilde{\inst}$
will contain two sets of variables $\{y_i\}_{i\in [n]}$ and
$\{z_i\}_{i\in [n]}$, with every constraint having exactly one variable from
$\{y_i\}_{i\in [n]}$ and two variables from $\{z_j\}_{j\in [n]}$.  Notice
that if we set $G = \{y_i\}_{i\in [n]}$, then objective
(\ref{item:obj2}) holds vacuously, i.e., every constraint in $\tilde{\inst}$
is active.  The former objective
(\ref{item:obj1}) is achieved as a direct consequence of
anticoncentration of low degree polynomials (see \pref{fact:anticonc}). In the case of \klin{k}, the second objective is achieved by slightly modifying the greedy assignment: we flip each of the assignments for the greedy variables with a small probability $\eta$ (that corresponds to one of the extrema of the degree-$k$ Chebyshev polynomials of the first kind).

Our algorithm for \emph{triangle-free} instances begins by
picking $(F,G)$ to be a random partition of the variables.
In this case, after fixing a random assignment to $F$, a natural greedy
strategy would proceed as follows: Assign each variable in $G$ a value that
satisfies the maximum the number of its own active constraints.

In order to achieve objective (\ref{item:obj2}), it is sufficient if for each
inactive constraint its variables are assigned independently and uniformly at random.
Since the instance is \emph{triangle-free},
for every pair of variables $x_i,x_j \in G$ the
active constraints of $x_i$ and $x_j$ are over disjoint sets of
variables.  This implies that the greedy assignments for variables
within each inactive constraint are already independent.
Unfortunately, the greedy assignment as defined above could
possibly be biased, and in general much worse than a random
assignment on the inactive constraints. We overcome this technical
hurdle by using a modified greedy strategy defined as follows.  Assign
$-1$ to all variables in $G$ and then for
each variable $x_i \in G$, consider the change in the number of active
constraints satisfied if we flip $x_i$ from $-1$ to $1$.  The algorithm
will flip the value only if this number exceeds an appropriately
chosen threshold $\theta_i$.  The threshold $\theta_i$ is chosen so as to ensure
that over all choices of values to $F$, the assignment to $x_i$ is unbiased.
Triangle-freeness implies that these
assignments are independent within each inactive constraint.  Putting
these ideas together, we obtain the algorithm for triangle-free
instances discussed in \pref{sec:trianglefree}.


%% file: content/preliminaries.tex
\section{Preliminaries}
\label{sec:preliminaries}


\paragraph{Constraint satisfaction problems}
We will be considering a somewhat general form of constraint satisfaction problems.
An instance for us will consist of~$n$ Boolean variables and~$m$ constraints.
We call the variables $x_1, \dots, x_n$, and we henceforth think of them as
taking the Boolean values $\pm 1$.  Each constraint is a pair $(P_\ell, S_\ell)$
(for $\ell \in [m]$) where $P_\ell : \sbits^r \to \{0,1\}$ is the \emph{predicate},
and $S_\ell$ is the \emph{scope}, an ordered $r$-tuple of distinct coordinates
from~$[n]$.  The associated constraint is that $P_\ell(x_{S_\ell}) = 1$, where
we use the notation $x_{S}$ to denote variables $x$ restricted to coordinates~$S$.
We always assume (without loss of generality) that $P_\ell$ depends on all~$r$
coordinates.  The number $r$ is called the \emph{arity} of the constraint, and
throughout this paper $k$ will denote an upper bound on the arity of all
constraints.  Typically we think of~$k$ as a small constant.

We are also interested in the special case of \klin{k}.  By this we mean the case when all constraints are XORs of exactly $k$~variables or their negations; in other words, when every~$P_\ell$ is of the form $P_\ell(x_1, \dots, x_k) = \frac12 \pm \frac12 x_1 x_2 \cdots x_k$.  When discussing \klin{k} we will also always make the assumption that all scopes are distinct as sets; i.e., we don't have the same constraint or its negation more than once.

\paragraph{Hypergraph structure}  We will be particularly interested in the \emph{degree} $\deg(i)$ of each variable~$x_i$ in an instance.  This is simply the number of constraints in which $x_i$ participates; i.e., $\#\{\ell : S_\ell \ni i\}$.  Throughout this work, we let $\degree$ denote an upper bound on the degree of all variables.

For our second theorem, we will need to define the notion of ``triangle-freeness''.
\begin{definition} \label{def:trianglefree}
We say that an instance is \emph{triangle-free} if the scopes of any two distinct constraints intersect on at most one variable (``no overlapping constraints'') and, moreover, there are no three distinct constraints any two of whose scopes intersect (``no hyper-triangles''), see Figure~\ref{fig:triangle-free}.
\end{definition}

\begin{figure}
\begin{center}
\includegraphics{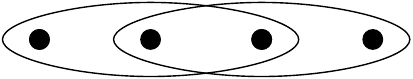}
\qquad
\includegraphics{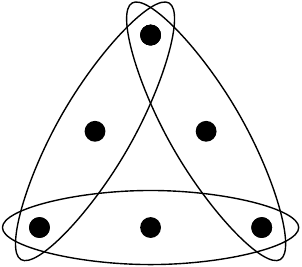}
\end{center}
\caption{The two forbidden configurations for triangle-free instances}
\label{fig:triangle-free}
\end{figure}

\paragraph{Fourier representation} We recall that any Boolean function $f : \sbits^n \to \R$ can be represented by a multilinear polynomial, or \emph{Fourier expansion},
\[
    f(x) = \sum_{S \subset [n]} \wh{f}(S) x^S, \quad \text{ where } x^S \defeq \prod_{i \in S} x_i.
\]
For more details see, e.g.,~\cite{Ryanbook}; we recall here just a few facts we'll need.  First, $\E[f(\bx)] = \wh{f}(\emptyset)$.  (Here and throughout we use \textbf{boldface} for random variables; furthermore, unless otherwise specified $\bx$ refers to a uniformly random Boolean string.)  Second, Parseval's identity is $\|f\|_2^2 = \E[f(\bx)^2] = \sum_S \wh{f}(S)^2$, from which it follows that $\Var[f(\bx)] = \sum_{S \neq \emptyset} \wh{f}(S)^2$.  Third,
\[
    \Inf_i[f] = \sum_{S \ni i} \wh{f}(S)^2 = \E[(\partial_i f)(\bx)^2],
\]
where $\partial_i f$ is the \emph{derivative} of~$f$ with respect to the $i$th coordinate.  This can be defined by the factorization $f(x) = x_i \cdot (\partial_i f)(x') + g(x')$, where $x' = (x_1, \dots, x_{i-1}, x_{i+1}, \dots, x_n)$, or equivalently by $\partial_i f(x') = \frac{f(x',+1) - f(x',-1)}{2}$, where here $(x',b)$ denotes $(x_1, \dots, x_{i-1}, b, x_{i+1}, \dots, x_n)$.\ROnote{Do we want to make it a lemma that, like, for Max-$k$XOR the $i$th influence is directly related to $\deg(i)$?}  We record here a simple fact about these derivatives:
\begin{lemma} \label{lem:highdegree}
  For any predicate $P : \sbits^r \to \{0,1\}$,  $r \geq 2$, we have $\Var[(\partial_i P)(\bx)] \geq \Omega(2^{-r})$ for all $i$.
\end{lemma}
\begin{proof}
The function $\partial_i P(x)$ takes values in $\{-\frac12, 0, \frac12\}$.
It cannot be constantly~$0$, since we assume~$P$ depends on its $i$th input.  It also cannot be constantly $\frac12$, else we would have~$P(x) = \frac12 + \frac12 x_i$ and so~$P$ would not depend on all $r \geq 2$ coordinates.  Similarly it cannot be constantly $-\frac12$.  Thus $\partial_i P(x)$ is nonconstant, so its variance is~$\Omega(2^{-r})$.
\end{proof}

Given an instance and an assignment $x \in \sbits^n$,
the number of constraints satisfied by the assignment is simply $\sum_{\ell} P_\ell(x_{S_\ell})$.  This can be thought of as a multilinear polynomial
$\sbits^n \to \R$
of degree\footnote{We have the usual unfortunate terminology clash; here we mean degree as a polynomial.} at most~$k$.  We would like to make two minor adjustments to it, for simplicity.  First, we will normalize it by a factor of $\frac{1}{m}$ so as to obtain the \emph{fraction} of satisfied constraints.  Second, we will replace $P_\ell$ with $\ol{P}_\ell$, defined by
\[
    \ol{P}_\ell = P_\ell - \E[P_\ell] = P_\ell - \wh{P_\ell}(\emptyset).
\]
In this way, $\ol{P}_\ell(x_{S_\ell})$ represents the \emph{advantage} over a random assignment.  Thus given an instance, we define the \emph{associated polynomial} $\clf(x)$ by
\[
    \clf(x) = \frac{1}{m} \sum_{\ell = 1}^m \ol{P}_\ell(x_{S_\ell}).
\]
This is a polynomial of degree at most~$k$ whose value on an assignment~$x$ represents the advantage obtained over a random assignment in terms of the fraction of constraints satisfied.  In general, the algorithms in this paper are designed to find assignments $x \in \sbits^n$ with $\clf(x) \geq \Omega(\frac{1}{\sqrt{\degree}})$.\ROnote{If we wanted, here is where we might link influence to $\degree$.}

\paragraph{Low-degree polynomials often achieve their expectation}  Our proofs will frequently rely on the following fundamental fact from Fourier analysis, whose proof depends on the well-known ``hypercontractive inequality''.  A proof of this fact appears in, e.g.,~\cite[Theorem 9.24]{Ryanbook}.
\begin{fact} \label{fact:anticonc}
Let $f:\sbits^n \rightarrow \R$ be a multilinear polynomial of degree at most $k$.
Then $\Pr[f(\bx) \geq \E[f]] \ge \frac{1}{4} \exp(-2k)$.
In particular, by applying this to $f^2$, which has degree at most $2k$, we get
\[
    \Pr\Bigl[|f(\bx)| \geq \|f\|_2\Bigr] \geq \exp(-O(k))
\]
which implies that
\[
   \E\Bigl[|f(\bx)|\Bigr] \geq \exp(-O(k)) \cdot \|f\|_2 \geq \exp(-O(k)) \cdot \stddev[f(\bx)]
	\; .
\]
\end{fact}

%% file: content/3xor.tex
\section{A simple proof for \texorpdfstring{Max-$\boldsymbol{3}$XOR}{Max-3XOR}} \label{sec:3xor}
We begin by proving \pref{thm:kxor} in the case of \klin{3}, as the proof can be somewhat streamlined in this case.
Given an instance of \klin{3} we have the corresponding polynomial
\[
    \clf(x) = \sum_{|S| = 3} \wh{\clf}(S) x^S = \sum_{i,j,k \in [n]} a_{ijk} x_ix_jx_k,
\]
where $\wh{\clf}(S) \in \{\pm \frac{1}{2m}, 0\}$ depending on whether the corresponding constraint exists in the instance, and where we have introduced $a_{ijk} = \tfrac16 \wh{\clf}(\{i,j,k\})$ for $i,j,k \in [n]$ distinct. We now use the trick of ``decoupling'' the first coordinate (cf.~\cite[Lem.~2.1]{KN08}); i.e., our algorithm will consider $\wt{\clf}(y,z) = \sum_{i,j,k} a_{ijk} y_i z_j z_k$, where $y_1, \dots, y_n, z_1, \dots, z_n$ are new variables. The algorithm will ultimately produce a good assignment $y,z \in \{\pm 1\}^n$ for~$\wt{\clf}$.  Then it will define an assignment $\bx \in \{\pm 1\}^n$ by using one of three ``randomized rounding'' schemes:
\[
    \text{w.p.\  $\tfrac49$, } \
                                        \bx_i = \begin{cases}
                                                   y_i & \text{w.p.\  $\frac12$}\\
                                                   z_i & \text{w.p.\  $\frac12$}
                                              \end{cases} \ \ \forall i; \qquad
    \text{w.p.\  $\tfrac49$, } \
                                        \bx_i = \begin{cases}
                                                   y_i & \text{w.p.\  $\frac12$}\\
                                                   -z_i & \text{w.p.\  $\frac12$}
                                              \end{cases} \ \ \forall i; \qquad
    \text{w.p.\  $\tfrac19$, }\
                                        \bx_i = -y_i \ \ \forall i.
\]
We have that $\E[\clf(\bx)]$ is equal to
\begin{align}
    &\tfrac49\littlesum_{i,j,k}a_{ijk}(\tfrac{y_i+z_i}{2})(\tfrac{y_j+z_j}{2})(\tfrac{y_k+z_k}{2})
+ \tfrac49\littlesum_{i,j,k}a_{ijk}(\tfrac{y_i-z_i}{2})(\tfrac{y_j-z_j}{2})(\tfrac{y_k-z_k}{2})
+\tfrac19\littlesum_{i,j,k}a_{ijk}(-y_i)(-y_j)(-y_k) \nonumber\\
    =\ &\tfrac19 \sum_{i,j,k}a_{ijk}(y_i z_j z_k + z_i y_j z_k + z_i z_j y_k) = \frac{1}{3} \wt{\clf}(y,z). \label{eqn:i-have-an-analogue}
\end{align}
Thus in expectation, the algorithm obtains an assignment for $\clf$ achieving at least $\frac13$ of what it achieves for $\wt{\clf}$.

Let us now write $\wt{\clf}(y,z) = \sum_i y_i G_i(z)$, where $G_i(z) = \sum_{j,k}a_{ijk}z_jz_k$.  It suffices for the algorithm to find an assignment for~$z$ such that $\sum_i |G_i(z)|$ is large, as it can then achieve this quantity by taking $y_i = \sgn(G_i(z))$. The algorithm simply chooses $\bz \in \{\pm 1\}^n$ uniformly at random.  By Parseval
we have $\E[G_i(\bz)^2] = \sum_{j<k} (2a_{ijk})^2 = \frac19\Inf_i[\clf]$ for each~$i$.  Applying \pref{fact:anticonc} (with $k=2$) we therefore get $\E[|G_i(\bz)|] \geq \Omega(1) \cdot \sqrt{\Inf_i[\clf]}$.  Since $\Inf_i[\clf] = \deg(i)/4m^2$, we conclude
\[
    \E\left[\littlesum_i |G_i(\bz)|\right] \geq \Omega(1) \cdot \littlesum_i \tfrac{\sqrt{\deg(i)}}{m} \geq \Omega(1) \cdot \littlesum_i \tfrac{\deg(i)}{m\sqrt{D}} = \Omega(1) \cdot \frac{1}{\sqrt{\degree}}.
\]
As $\littlesum_i |G_i(\bz)|$ is bounded by~$1/2$,\Onote{was: 1} Markov's inequality implies that the algorithm can with high probability find a~$z$ achieving $\sum_i |G_i(z)| \geq \Omega(\frac{1}{\sqrt{\degree}})$ after $O(\sqrt{\degree})$ trials of~$\bz$.  As stated, the algorithm then chooses $y$ appropriately to attain $\wt{\clf}(y,z) \geq \Omega(\frac{1}{\sqrt{\degree}})$, and finally gets $\frac13$ of this value (in expectation) for $\clf(x)$.  \ROnote{Shall we bother to point out we can upgrade this final expectation to high probability?  That we can derandomize the algorithm?  Perhaps `no' and `no'?}
\Bnote{Added a comment about derandomization in the intro}

\paragraph{Derandomization} It is easy to efficiently derandomize the above algorithm. The main step is to recognize that ``$(2,4)$-hypercontractivity'' is all that's needed for \pref{fact:anticonc} (perhaps with a worse constant); thus it holds even when the random bits are merely $4$-wise independent.  This is well known, but we could not find an explicit reference; hence we give the proof in the case when $f$ is homogeneous of degree~$2$ (the case that's needed in the above algorithm).  Without loss of generality we may assume $\E[f(\bx)] = 0$ and $\E[f(\bx)^2] = 1$.  Then it's a simple exercise to check that $\E[f(\bx)^4] \leq 15$, and this only requires the bits of~$\bx$ to be $4$-wise independent. But now
\[
    \Pr[f(\bx) \geq 0] = \E[1_{\{f(\bx) \geq 0\}}] \geq \E[.13f(\bx) + .06 f(\bx)^2 - .002f(\bx)^4] \geq .06 - .002 \cdot 15 = .03
\]
where we used the elementary fact $1_{\{t \geq 0\}} \geq .13 t +.06 t^2 -.002 t^4$ for all $t \in \R$. Thus indeed the algorithm can find a $z$ achieving $\sum_i |G_i(z)| \geq \Omega(\frac{1}{\sqrt{D}})$ by enumerating all strings in a $4$-wise independent set; it is well known this can be done in polynomial time.  Following this, the algorithm chooses string~$y$ deterministically.  Finally, it is clear that each of the three different randomized rounding schemes only requires $3$-wise independence, and a deterministic algorithm can simply try all three and choose the best one.

%% file: content/dfko.tex
\section{A general result for bounded-influence functions}
\label{sec:proof}

One can obtain our \pref{thm:kxor} for higher odd~$k$ by generalizing the proof in the preceding section.  Constructing the appropriate ``randomized rounding'' scheme to decouple the first variable becomes slightly more tricky, but one can obtain the identity analogous to~\eqref{eqn:i-have-an-analogue} through the use of Chebyshev polynomials.  At this point the solution becomes very reminiscent of the Dinur~et~al.~\cite{Dinuretal} work.  Hence in this section we will simply directly describe how one can make~\cite{Dinuretal} algorithmic.

The main goal of~\cite{Dinuretal} was to understand the ``Fourier tails'' of bounded degree-$k$ polynomials.  One of their key technical results was the following theorem, showing that if a degree-$k$ polynomial has all of its influences small, it must deviate significantly from its mean with noticeable probability:
\begin{theorem}                                     \label{thm:dfko3}
    (\cite[Theorem~3]{Dinuretal}.)  There is a universal constant $C$ such that the following holds.  Suppose $g : \sbits^n \to \R$ is a polynomial of degree at most~$k$ and assume $\Var[g] = 1$.  Let $t \geq 1$ and suppose that $\Inf_i[g] \leq C^{-k} t^{-2}$ for all~$i \in [n]$.  Then
    \[
        \Pr[|g(\bx)| \geq t] \geq \exp(-C t^2 k^2 \log k).
    \]
\end{theorem}
In the context of \klin{k}, this theorem already nearly proves our \pref{thm:kxor}.  The reason is that in this context, the associated polynomial $\clf(x)$ is given by
\[\clf(x)=\frac{1}{2m} \sum_{\ell=1}^m b_\ell \prod_{j \in S_\ell} x_j, ~ \text{where } b_\ell \in \{-1,1\}.\]
Hence $\Var[\clf] = 1/4m$ and $\Inf_i[\clf] = \deg(x_i)/4m^2 \leq \degree/4m^2$.
\ROnote{again, should this be a fact or a lemma or something?}  Taking $g = 2\sqrt{m} \cdot \clf$ and $t = \exp(-O(k)) \cdot \sqrt{m/\degree}$, \pref{thm:dfko3} immediately implies that
\begin{equation}    \label{eqn:dfko-guar}
    \Pr\Bigl[|\clf(\bx)| \geq \exp(-O(k)) \cdot \frac{1}{\sqrt{D}}\Bigr] \geq \exp(-O(m/\degree)).
\end{equation}
This already shows the desired existential result, that there \emph{exists} an assignment beating the random assignment by $\exp(-O(k)) \cdot \frac{1}{\sqrt{D}}$.  The only difficulty is that the low probability bound in~\eqref{eqn:dfko-guar} does not imply we can find such an assignment efficiently.

However this difficulty really only arises because~\cite{Dinuretal} had different goals.  In their work, it was essential to show that~$g$ achieves a slightly large value on a completely \emph{random} input.\footnote{Also, their efforts were exclusively focused on the parameter~$k$, with quantitative dependencies on~$t$ not mattering.  Our focus is essentially the opposite.}  By contrast, we are at liberty to show~$g$ achieves a large value however we like --- semi-randomly, greedily --- so long as our method is algorithmic.  That is precisely what we do in this section of the paper.  Indeed, in order to ``constructivize''~\cite{Dinuretal}, the only fundamental adjustment we need to make is at the beginning of the proof of their Lemma~1.3: when they argue that ``$\Pr[|\ell(\bx)| \geq t'] \geq \exp(-O({t'}^2))$ for the degree-$1$ polynomial~$\ell(x)$'', we can simply greedily choose an assignment~$x$ with $|\ell(x)| \geq t'$.

Our constructive version of \pref{thm:dfko3} follows.  It directly implies our \pref{thm:kxor}, as described above.

\def\AdvRand{\textsc{AdvRand}\xspace}

\begin{theorem}\label{thm:constructive}
    There is a universal constant $C$ and a randomized algorithm such that the following holds.  Let $g : \sbits^n \to \R$ be a polynomial with degree at most~$k$ and  $\Var[g] = 1$ be given.  Let $t \geq 1$ and suppose that $\Inf_i[g] \leq C^{-k} t^{-2}$ for all~$i \in [n]$.  Then with high probability the algorithm outputs an assignment~$x$ with $|g(x)| \geq t$.  The running time of the algorithm is $\poly(m,n,\exp(k))$, where $m$ is the number of nonzero monomials in~$g$.\footnote{For simplicity in our algorithm, we assume that exact real arithmetic can be performed efficiently.}
\ROnote{perhaps here is where we could discuss estimating the Chebyshev polynomials.  I know an exact citation for this if we want.  We're cheating here in an even more excruciatingly annoying way; we don't really have $\Var[g] = 1$, and if one insists on it by scaling then there is the potential for irrational square-roots to arise\dots}\AVnote{I don't think we need to talk about estimating issues. }
\end{theorem}

The algorithm \AdvRand achieving Theorem~\ref{thm:constructive} is given below.  It is derived directly from~\cite{Dinuretal}, and succeeds with probability that is inverse polynomial in $n$.
 The success probability is then boosted by running the algorithm multiple times.
We remark that $\eta_0^{(k)}, \eta_1^{(k)}, \dots, \eta_k^{(k)}$ denote the $k+1$ extrema in $[-1,1]$ of the $k$th Chebyshev polynomial of the first kind $T_k(x)$, and are given by $\eta_j^{(k)} = \cos(j \pi/k)$ for $0 \le j \le k$. We now describe the algorithm below, for completeness. In the rest of the section, we will assume without loss of generality that $k$ is odd (for even $k$, we just think of the polynomial as being of degree $k+1$, with the degree $(k+1)$ part being $0$). 

\vspace{10pt}

\def\compactify{\itemsep=0pt \topsep=0pt \partopsep=0pt \parsep=0pt}

\rule{0pt}{12pt}
\hrule height 0.8pt
\rule{0pt}{1pt}
\hrule height 0.4pt
\rule{0pt}{6pt}

\noindent \textbf{\AdvRand: Algorithm for Advantage over Average for degree $k$ polynomials}

\medskip

\noindent \textbf{Input:} a degree $k$-function $g$

\noindent \textbf{Output:} an assignment $x$

\begin{enumerate} \compactify
\item 
Let $1 \le s \le \log_2 k$ be a scale such that the mass (i.e., sum of squares of coefficients) of the Fourier transform of $g$ on levels between $2^{s-1}$ and $2^s$ is at least $1/\log k$.
\item For every $i \in [n]$,  put $i$ in set $U$ with probability $2^{-s}$. For every $i \notin U$, set $x_i \in \{-1,1\}$ uniformly at random and let $y$ be the assignment restricted to the variables in $[n] \setminus U$.
\item Let $\grest$ be the restriction obtained.
For every $j \in U$, set $x_j =\text{sign}(\gresth(\set{j}))$.
\item Pick $r \in \set{0,1, \dots,k}$ uniformly at random, and let $\eta=\eta_r^{(k)} /2$. 
\item For each coordinate $j \in U$, flip $x_j$ independently at random with probability $(1-\eta)/2$.
\item Output $x$.
\end{enumerate}

\hrule height 0.4pt
\rule{0pt}{1pt}
\hrule height 0.8pt
\rule{0pt}{12pt}

We now give the analysis of the algorithm, following~\cite{Dinuretal}. The second step of the algorithm performs a \emph{random restriction}, that ensures that $g_y$ has a lot of mass on the first-order Fourier coefficients. The key lemma (that follows from the proof of Lemma 1.3 and Lemma 4.1 in \cite{Dinuretal}) shows that we can find an assignment that obtains a large value for a polynomial with sufficient ``smeared'' mass on the first-order Fourier coefficients.

\begin{lemma}\label{lem:linear}
Suppose $g:\pmo^N\rightarrow \R$ has degree at most $k$, $t \ge 1$, and $\sum_{i \in [N]} \abs{\widehat{g}(\set{i})} \ge 2t(k+1)$. 
Then a randomized polynomial time algorithm
outputs a distribution 
over assignments $\bx \in \{-1,1\}^N$ such that
\[\Pr_{\bx}\left[|g(\bx)|\ge t\right] \ge \exp\left(-O(k)\right) .\]
\end{lemma}
\noindent The algorithm proving \pref{lem:linear} corresponds to Steps (3-6) of the Algorithm \AdvRand.
\begin{proof}
We sketch the proof, highlighting the differences to Lemma 1.3 of \cite{Dinuretal}.
First we observe that by picking the assignment $x^*_i= \text{sign}(\widehat{g}(\set{i}))$, we can maximize the linear portion as
\[\sum_{i \in [N]} \widehat{g}(\set{i})  x^*_i = \sum_{i \in [N]} \abs{\widehat{g}(\set{i})}\ge 2t(k+1).\]
From this point on, we follow the proof of Lemma 1.3 in \cite{Dinuretal} with their initial point $x_0$ being set to $x^*$.
Let $\bz \leftarrow_{\eta} \pmo^N$ be a random string generated by independently setting each coordinate ${\bz_j} =-1$ with probability $(1-\eta)/2$ (as in step 5 of the algorithm), and let
\[(T_\eta g)(x^*)= \E_{\bz \leftarrow_{\eta}\pmo^n} [g(x^* \cdot \bz)].\]
Lemma 1.3 of \cite{Dinuretal}, by considering 
$(T_\eta g)(x^*)$ as a polynomial in $\eta$ and using the extremal properties of Chebyshev polynomials (Corollary 2.8 in \cite{Dinuretal}),  shows that there exists 
$\eta \in \set{\tfrac{\eta^{(k)}_0}{2},\tfrac{\eta^{(k)}_1}{2}, \dots,\tfrac{\eta^{(k)}_k}{2}}$ such that
\begin{equation} \label{eq:lemcheby}
\E_{\bz \leftarrow_{\eta}\pmo^n} \Big[\abs{g(x^* \cdot \bz)} \Big] \ge 2t(k+1)\cdot\frac{1}{(2k+2)} = t.
\end{equation}

Consider $g(x^* \cdot \bz)$ as a polynomial in $\bz$, with degree at most $k$. As in \cite{Dinuretal}, we will now use the hypercontractivity to give a lower bound on the probability (over random $\bz$) that $\abs{g(x^* \cdot \bz)}$
exceeds the expectation. Note that our choice of $\eta \in [-\tfrac{1}{2}, \tfrac{1}{2}]$ and hence the bias is in the interval $[\tfrac{1}{4} , \tfrac{3}{4}]$. Using Lemma 2.5 in \cite{Dinuretal} (the analogue of \pref{fact:anticonc} for biased measures),  it follows that 
\[ \Pr_{\bz} \Big[ \abs{g(x^* \cdot \bz)}  \ge t\Big] \ge \tfrac{1}{4}  \exp(-2k). \]
Hence when $\bx$ is picked according to $\mathcal{D}$, 
with probability at least $1/(k+1)$ the algorithm chooses an $\eta$ such that \pref{eq:lemcheby} holds, and then a random $\bz$ succeeds with probability $\exp(-O(k))$, thereby giving the required success probability.
\end{proof}

We now sketch the proof of the constructive version of Theorem 3 in \cite{Dinuretal}, highlighting why algorithm \AdvRand works.
\begin{proof}[Proof of \pref{thm:constructive}]

The scale $s$ is chosen such that the Fourier coefficients of $g$ of order $[2^{s-1}, 2^{s}]$ have mass at least $1/\log k$. The algorithm picks set $U$ randomly by choosing each variable with probability $2^{-s}$, and $\grest$ is the restriction of $g$ to the coordinates in $U$ obtained by setting the other variables randomly to $\by \in \{-1,1\}^{[N]\setminus U}$.

Let $\gamma_i= \sum_{S \cap U=\set{i}} \widehat{g}(S)^2$.
Fixing $U$ and $y$, let the indices $T=\set{i \in U: \widehat{g}_y (\set{i})^2 \le (2e)^{2k} \gamma_i }$. The proof of Theorem 3 in \cite{Dinuretal} shows that a constant fraction of the first order Fourier coefficients are large; in particular after Steps 1 and 2 of the algorithm,
\begin{equation}\label{eq:sumhypercontractivity}
\Pr_{U,\by} \Big[ \sum_{i \in T} \widehat{g}_{\by}(\set{i})^2 \ge \frac{1}{100 \log k} \Big] \ge \exp(-O(k)) \; .
\end{equation}
Further, for $i \in T$, we have $\abs{\widehat{g}_y (\set{i})}\le (2e)^k \sqrt{\gamma_i} \le (2e)^k \sqrt{\Infl_i(g)}$. Hence, when the above event in \pref{eq:sumhypercontractivity} is satisfied we have 
\begin{align*}
\sum_{i \in U} \abs{\widehat{g}_y (\set{i})} &\ge \frac{1}{\max_{i \in T} \abs{\widehat{g}_y (\set{i})}}\cdot \sum_{i \in T} \widehat{g}_y (\set{i})^2 \\
& \ge \frac{1}{(2e)^k \sqrt{\max_i \Infl_i (g)}} \cdot \frac{1}{100 \log k} \ge 2t(k+1). 
\end{align*}
Hence, applying \pref{lem:linear} with $g_y$ we get that
\begin{equation}
\Pr_{\bx \in \mathcal{D}} \Big[ \abs{g(x)} \ge t \Big] \ge \exp(-O(k)),
\end{equation}
where $\mathcal{D}$ is the distribution over assignments $x$ output by the algorithm.
Repeating this algorithm $\exp(O(k))$ times, we get the required high probability of success.
\end{proof}


%% file: content/trianglefree.tex
\section{Triangle-free instances}
\label{sec:trianglefree}

In this section we present the proof of \pref{thm:maingencsp}, which gives an efficient algorithm for beating the random assignment in the case of arbitrary triangle-free CSPs (recall \pref{def:trianglefree}).
We now restate \pref{thm:maingencsp} and give its proof. As in the proof of \pref{thm:constructive}, we can easily move from an expectation guarantee to a high probability guarantee by first applying Markov's inequality, and then repeating the algorithm $\exp(k) \poly(n,m)$ times; hence we will prove the expectation guarantee here. 
\begin{theorem}
  There is a $\poly(m,n,\exp(k))$-time randomized algorithm with the following guarantee.  Let the input be a triangle-free instance over $n$ Boolean variables, with $m$ arbitrary constraints each of arity between~$2$ and~$k$.   Assume  that each variable participates in at most~$\degree$ constraints.
  Let the associated polynomial be $\clf(x)$.  Then the algorithm outputs an assignment~$\bx \in \{\pm 1\}^n$ with
  \[
    \E[\clf(\bx)] \geq \exp(-O(k)) \cdot \sum_{i=1}^n \frac{\sqrt{\deg(i)}}{m} \geq \exp(-O(k)) \cdot \frac{1}{\sqrt{\degree}}.
  \]
\end{theorem}
\begin{proof}
    Let $(F,G)$ be a partition of~$[n]$, with $F$ standing for ``Fixed'' and $G$ standing for ``Greedy''. Eventually the algorithm will choose the partition randomly, but for now we treat it as fixed.  We will write the two parts of the algorithm's random assignment $\bx$ as $(\bx_F,\bx_G)$.  The bits $\bx_F$ will first be chosen independently and uniformly at random.  Then the bits~$\bx_G$ will be chosen in a careful way which will make them uniformly random, but not completely independent.

    To make this more precise, define a constraint $(P_\ell, S_\ell)$ to be \emph{active} if its scope $S_\ell$ contains exactly one coordinate from~$G$.  Let us partition these active constraints into groups
    \[
        N_j = \{ \ell : (P_\ell,S_\ell) \text{ is active and } S_\ell \ni j\}, \quad j \in G.
    \]
    For each coordinate $j \in G$, we'll define $A_j \subset F$ to be the union of all active scopes involving~$j$ (but excluding $j$ itself); i.e.,
 \[
        A_j = \bigcup \{S_\ell \setminus \{j\} : \ell \in N_j \}.
    \]
      This set $A_j$ may be empty.  Our algorithm's choice of $\bx_G$ will have the following property:
    \begin{center}
        \emph{$\forall j \in G$, the distribution of $\bx_j$ is uniformly random, and it depends only on $(\bx_{i} : i \in A_j)$.} \ \  $(\dagger)$
    \end{center}
    From property~$(\dagger)$ we may derive:
    \begin{subclaim}       \label{claim:indep}
       For every \emph{inactive} constraint $(P_\ell,S_\ell)$, the random assignment bits $\bx_{S_\ell}$ are uniform and \emph{independent}.
    \end{subclaim}
    \begin{proof}[Proof of Claim.]
        First consider the coordinates $j \in S_\ell \cap G$. By the property~$(\dagger)$, each such~$\bx_j$ depends only on $(\bx_{i} : i \in A_j)$; further, these sets $A_j$ are disjoint precisely because of the ``no hyper-triangles'' part of triangle-freeness.  Thus indeed the bits $(\bx_j : j \in S_\ell \cap G)$ are uniform and mutually independent.  The remaining coordinates $S_\ell \cap F$ are also disjoint from all these $(A_j)_{j \in S_\ell \cap G}$, by the ``no overlapping constraints'' part of the triangle-free property.  Thus the remaining bits  $(\bx_i : i \in S_\ell \cap F)$ are uniform, independent, and independent of the bits $(\bx_j : j \in S_\ell \cap G)$, completing the proof of the claim.
    \end{proof}
    An immediate corollary of the claim is that all inactive constraints, $\ol{P}_\ell$ contribute nothing, in expectation, to $\E[\clf(\bx)]$. Thus it suffices to consider the contribution of the active constraints.  
    Our main goal will be to show that the bits $\bx_G$ can be chosen in such a way that
    \begin{equation}        \label{eqn:active-contribution}
        \forall j \in G \quad \E\Bigl[\littlesum_{\ell \in N_j} \ol{P}_\ell(\bx_{S_\ell})\Bigr] \geq \exp(-O(k)) \cdot \sqrt{|N_j|}
    \end{equation}
    and hence
    \begin{equation}        \label{eqn:active-contribution-cor}
        \E[\clf(\bx)] \geq \frac{1}{m} \cdot \exp(-O(k)) \cdot \sum_{j \in G} \sqrt{|N_j|}.
    \end{equation}
    Given~\eqref{eqn:active-contribution-cor} it will be easy to complete the proof of the theorem by choosing the partition $(F,G)$ randomly.

    So towards showing~\eqref{eqn:active-contribution}, fix any $j \in G$. For each $\ell \in N_j$ we can write $\ol{P}_\ell(x_{S_\ell}) = x_j Q_\ell(x_{S_\ell \setminus \{j\}}) + R_\ell(x_{S_\ell \setminus \{j\}})$, where $Q_\ell = \partial_j \ol{P}_\ell = \partial_j P_\ell$.  Since the bits $\bx_i$ for $i \in S_\ell \setminus \{j\} \subset F$ are chosen uniformly and independently, the expected contribution to~\eqref{eqn:active-contribution} from the $R_{\ell}$ polynomials is~$0$.  Thus we just need to establish
    \begin{equation} \label{eqn:active-contrib2}
        \E\Bigl[\bx_{j} \cdot \littlesum_{\ell \in N_j} \bQ_\ell \Bigr] \geq \exp(-O(k)) \cdot \sqrt{|N_j|}, \quad \text{ where } \bQ_\ell \defeq Q_\ell(\bx_{S_\ell \setminus \{j\}}).
    \end{equation}
    We now finally describe how the algorithm chooses the random bit $\bx_j$.  Naturally, we will choose it to be $+1$ when $\littlesum_{\ell \in N_j} \bQ_\ell$ is ``large'' and $-1$ otherwise. Doing this satisfies the second aspect of property~$(\dagger)$, that $\bx_j$ should depend only on $(\bx_i : i \in A_j)$.  To satisfy the first aspect of property~$(\dagger)$, that $\bx_j$ is equally likely $\pm 1$, we are essentially forced to define
    \begin{equation}        \label{eqn:triangle-x-def}
        \bx_j = \sgn\Bigl(\littlesum_{\ell \in N_j} \bQ_\ell - \theta_j\Bigr),
    \end{equation}
    where $\theta_j$ is defined to be a \emph{median} of the random variable $\littlesum_{\ell \in N_j} \bQ_\ell$.

    (Actually, we have to be a little careful about this definition.  For one thing, if the median $\theta_j$ is sometimes achieved by the random variable, we would have to carefully define $\sgn(0)$ to be sometimes $+1$ and sometimes $-1$ so that $\bx_j$ is equally likely $\pm 1$.  For another thing, we are assuming here that the algorithm can efficiently \emph{compute} the medians $\theta_j$.  We will describe how to handle these issues in a technical remark after the proof.)

    Having described the definition~\eqref{eqn:triangle-x-def} of $\bx_j$ satisfying property~$(\dagger)$, it remains to verify the inequality~\eqref{eqn:active-contrib2}. Notice that by the ``no overlapping constraints'' aspect of triangle-freeness, the random variables $\bQ_\ell$ are actually mutually independent.  Further, \pref{lem:highdegree} implies that each has variance $\Omega(2^{-k})$; hence the variance of $\bQ \defeq \sum_{\ell \in N_j} \bQ_\ell$ is $\exp(-O(k)) \cdot |N_j|$.  
		\Onote{consider definining $Q$ earlier}
		Thus inequality~\eqref{eqn:active-contrib2} is equivalent to
    \[
        \E[\sgn(\bQ - \theta_j) \bQ] \geq \exp(-O(k)) \cdot \stddev[\bQ] = \exp(-O(k)) \cdot \stddev[\bQ - \theta_j].
    \]
    Now
    \begin{equation}    \label{eqn:what-if-not-0}
        \E[\sgn(\bQ - \theta_j) \bQ] = \E[\sgn(\bQ - \theta_j) (\bQ-\theta_j + \theta_j)] = \E[\left|\bQ - \theta_j\right|] + \E[\bx_j \cdot \theta_j].
    \end{equation}
    We have $\E[\bx_j \cdot \theta_j] = 0$ since $\E[\bx_j] = 0$.  And as for $\E[\left|\bQ - \theta_j\right|]$, it is indeed at least~$\exp(-O(k)) \cdot \stddev[\bQ]$ by \pref{fact:anticonc}, since~$\bQ$ is a degree-$(k-1)$ function of uniform and independent random bits. \ROnote{I know this is overkill because $\bQ$ is ``effectively degree-$1$'', but anyway we already have this Fact} Thus we have finally established~\eqref{eqn:active-contribution}, and therefore~\eqref{eqn:active-contribution-cor}.

    To conclude, we analyze what happens when the algorithm initially chooses a uniformly \emph{random} partition $(\bF,\bG)$ of $[n]$.  In light of~\eqref{eqn:active-contribution-cor}, it suffices to show that for each $i \in [n]$ we have
    \begin{equation} \label{eqn:finish-triangle-free}
        \E\left[{\mathbf 1}[i \in \bG] \cdot \sqrt{|\bN_i|}\right] \geq \exp(-O(k)) \cdot \sqrt{\deg(i)}.
    \end{equation}
    We have $\Pr[i \in \bG] = \frac12$; conditioning on this event, let us consider the random variable~$|\bN_i|$; i.e., the number of active constraints involving variable~$x_i$. A constraint scope $S_\ell$ containing~$i$ becomes active if and only if all the other indices in $S_\ell$ go into~$\bF$, an event that occurs with probability  $2^{-k+1}$ (at least).  Furthermore, these events are independent across the scopes containing~$i$ because of the ``no overlapping constraints'' property of triangle-freeness.  Thus (conditioned on $i \in \bG$), each random variable $|\bN_i|$ is the sum $\bA_1 + \cdots + \bA_{\deg(i)}$ independent indicator random variables, each with expectation at least $2^{-k+1}$.  Thus we indeed have $\E[\sqrt{|\bN_i|}] \geq \exp(-O(k)) \sqrt{\deg(i)}$ as needed to complete the proof of~\eqref{eqn:finish-triangle-free}. This follows from the well known fact that $\E[\sqrt{\text{Binomial}(d,p)}] \geq \Omega(\min(\sqrt{dp}, dp))$.  (Alternatively, this follows from the fact that $\bA_1 + \cdots + \bA_{d_i}$ is at least its expectation $d_i 2^{-k+1}$ with probability at least $\exp(-O(k))$, by \pref{fact:anticonc}.  Here we would use that the $\bA_j$'s are degree-$(k-1)$ functions of independent random bits defining $(\bF,\bG)$).  The proof is complete.
\end{proof}

\begin{remark}  Regarding the issue of algorithmically obtaining the medians in the above proof: In fact, we claim it is unnecessary for the algorithm to compute the median $\theta_j$ of each $\bQ_j$ precisely.  Instead, our algorithm will (with high probability) compute a number~$\wt{\theta}_j$ and a probabilistic way of defining $\sgn(0) \in \{\pm 1\}$ such that, when $\bx_j$ is defined to be $\sgn(\bQ - \wt{\theta}_j)$, we have $\left|\E[\bx_j]\right| \leq \delta$, where $\delta = 1/\poly(m,n,\exp(k))$ is sufficiently small.  First, let us briefly say why this is sufficient.  The above proof relied on $\E[\bx_j] = 0$ in two places.  One place was in the last term of~\eqref{eqn:what-if-not-0}, where we used $\E[\bx_j \cdot \theta_j] = 0$.  Now in the approximate case, we'll have $|\E[\bx_j \cdot \wt{\theta}_j]| \leq \delta m$, and by taking $\delta$ appropriately small this will contribute negligibly to the overall theorem.  The other place that $\E[\bx_j] = 0$ was used was in deducing from \pref{claim:indep}, that the inactive constraints contributed nothing to the algorithm's expected value.  When we merely have $\left|\E[\bx_j]\right| \leq \delta$ (but still have the independence used in the claim), it's easy to see from Fourier considerations that each inactive constraint still contributes at most~$2^k \delta$ to the overall expectation, and again this is negligible for the theorem as a whole if $\delta = 1/\poly(m,n,\exp(k))$ is sufficiently small.  Finally, it is not hard to show that the algorithm can compute an appropriate $\wt{\theta}_j$ and probabilistic definition of $\sgn(0)$ in $\poly(m,n,\exp(k))$ time (with high probability), just by sampling to find a good approximate median $\wt{\theta}_j$ and then also estimating $\Pr[\bQ_j = \wt{\theta}_j]$ to handle the definition of $\sgn(0)$.
\end{remark}
\AVnote{I think one can do a dynamic program to calculate this in time $2^k \poly(n)$, as Oded suggested. Ignoring this for now. }

%% file: lowdegcsp.bbl
\def\cprime{$'$} \def\cprime{$'$}
  \def\ocirc#1{\ifmmode\setbox0=\hbox{$#1$}\dimen0=\ht0 \advance\dimen0
  by1pt\rlap{\hbox to\wd0{\hss\raise\dimen0
  \hbox{\hskip.2em$\scriptscriptstyle\circ$}\hss}}#1\else {\accent"17 #1}\fi}
  \def\cprime{$'$} \def\cprime{$'$} \def\cprime{$'$} \def\cprime{$'$}
  \def\cprime{$'$} \def\cprime{$'$} \def\cprime{$'$} \def\cprime{$'$}
  \def\cprime{$'$} \def\polhk#1{\setbox0=\hbox{#1}{\ooalign{\hidewidth
  \lower1.5ex\hbox{`}\hidewidth\crcr\unhbox0}}}
  \def\cfac#1{\ifmmode\setbox7\hbox{$\accent"5E#1$}\else
  \setbox7\hbox{\accent"5E#1}\penalty 10000\relax\fi\raise 1\ht7
  \hbox{\lower1.15ex\hbox to 1\wd7{\hss\accent"13\hss}}\penalty 10000
  \hskip-1\wd7\penalty 10000\box7} \def\cprime{$'$} \def\cprime{$'$}
  \def\cprime{$'$} \def\cprime{$'$} \def\cprime{$'$} \def\cprime{$'$}
  \def\cprime{$'$} \def\cprime{$'$} \def\cprime{$'$} \def\cprime{$'$}
  \def\cprime{$'$} \def\cprime{$'$} \def\cprime{$'$}
  \def\ocirc#1{\ifmmode\setbox0=\hbox{$#1$}\dimen0=\ht0 \advance\dimen0
  by1pt\rlap{\hbox to\wd0{\hss\raise\dimen0
  \hbox{\hskip.2em$\scriptscriptstyle\circ$}\hss}}#1\else {\accent"17 #1}\fi}
  \def\cprime{$'$} \def\cprime{$'$} \def\cprime{$'$} \def\cprime{$'$}
\providecommand{\bysame}{\leavevmode\hbox to3em{\hrulefill}\thinspace}
\providecommand{\MR}{\relax\ifhmode\unskip\space\fi MR }
\providecommand{\MRhref}[2]{%
  \href{http://www.ams.org/mathscinet-getitem?mr=#1}{#2}
}
\providecommand{\href}[2]{#2}
\begin{thebibliography}{DFKO07}

\bibitem[Alo97]{alon}
Noga Alon, \emph{On the edge-expansion of graphs}, Combin. Probab. Comput.
  \textbf{6} (1997), no.~2, 145--152.

\bibitem[BS90]{BergerShor}
Bonnie Berger and Peter~W. Shor, \emph{Approximation alogorithms for the
  maximum acyclic subgraph problem}, Proceedings of the First Annual ACM-SIAM
  Symposium on Discrete Algorithms (Philadelphia, PA, USA), SODA '90, Society
  for Industrial and Applied Mathematics, 1990, pp.~236--243.

\bibitem[DFKO07]{Dinuretal}
Irit Dinur, Ehud Friedgut, Guy Kindler, and Ryan O'Donnell, \emph{On the
  {F}ourier tails of bounded functions over the discrete cube}, Israel J. Math.
  \textbf{160} (2007), 389--412.

\bibitem[FGG14]{FGG14}
Edward Farhi, Jeffrey Goldstone, and Sam Gutmann, \emph{A quantum approximate
  optimization algorithm applied to a bounded occurrence constraint problem},
  2014, arXiv:1412.6062.

\bibitem[GZ12]{GuruswamiZhou}
Venkatesan Guruswami and Yuan Zhou, \emph{Approximating bounded occurrence
  ordering csps}, Approximation, Randomization, and Combinatorial Optimization.
  Algorithms and Techniques (Anupam Gupta, Klaus Jansen, José Rolim, and Rocco
  Servedio, eds.), Lecture Notes in Computer Science, vol. 7408, Springer
  Berlin Heidelberg, 2012, pp.~158--169 (English).

\bibitem[H{\aa}s00]{MR1761191}
Johan H{\aa}stad, \emph{On bounded occurrence constraint satisfaction}, Inform.
  Process. Lett. \textbf{74} (2000), no.~1-2, 1--6.

\bibitem[H{\aa}s01]{MR2144931}
\bysame, \emph{Some optimal inapproximability results}, J. ACM \textbf{48}
  (2001), no.~4, 798--859.

\bibitem[HV04]{HV04}
Johan H{\aa}stad and S.~Venkatesh, \emph{On the advantage over a random
  assignment}, Random Structures Algorithms \textbf{25} (2004), no.~2,
  117--149.

\bibitem[KN08]{KN08}
Subhash Khot and Assaf Naor, \emph{Linear equations modulo 2 and the {$L\sb 1$}
  diameter of convex bodies}, SIAM J. Comput. \textbf{38} (2008), no.~4,
  1448--1463.

\bibitem[Mak13]{Makarychev13}
Konstantin Makarychev, \emph{Local search is better than random assignment for
  bounded occurrence ordering k-csps}, 30th International Symposium on
  Theoretical Aspects of Computer Science, {STACS} 2013, February 27 - March 2,
  2013, Kiel, Germany, 2013, pp.~139--147.

\bibitem[O'D14]{Ryanbook}
Ryan O'Donnell, \emph{Analysis of {B}oolean functions}, Cambridge University
  Press, 2014.

\bibitem[She92]{shearer}
James~B. Shearer, \emph{A note on bipartite subgraphs of triangle-free graphs},
  Random Structures Algorithms \textbf{3} (1992), no.~2, 223--226.

\end{thebibliography}
